\documentclass[aps,pra,showpacs,twocolumn,superscriptaddress]{revtex4-1}
\usepackage{amsmath}
\usepackage{amsfonts}
\usepackage{graphicx}
\usepackage{dcolumn}
\usepackage{amsbsy}
\usepackage{bm}
\usepackage{amsthm}
\usepackage[caption=false]{subfig}
\usepackage{tabularx}
\usepackage{array}
\usepackage{natbib}
\usepackage{braket}
\usepackage{hhline}
\usepackage{verbatim}
\usepackage[colorlinks,
            linkcolor=blue,
            anchorcolor=blue,
            citecolor=blue,
	urlcolor=blue
            ]{hyperref}
\newtheorem{theorem}{Theorem}
\newtheorem{corollary}{Corollary}[theorem]

\makeatletter
\newcommand{\thickhline}{
    \noalign {\ifnum 0=`}\fi \hrule height 1pt
    \futurelet \reserved@a \@xhline
}
\newcolumntype{"}{@{\hskip\tabcolsep\vrule width 1pt\hskip\tabcolsep}}
\makeatother

\begin{document}
\title{Deterministic remote preparation of an arbitrary qubit state using a partially entangled state and finite classical communication}

\author{Congyi Hua}
 \affiliation{Zhejiang Institute of Modern Physics, Zhejiang University, Hangzhou 310027, China}
\author{Yi-Xin Chen}
 \email{yixinchenzimp@zju.edu.cn}
 \affiliation{Zhejiang Institute of Modern Physics, Zhejiang University, Hangzhou 310027, China}

\begin{abstract}
We propose a deterministic remote state preparation (RSP) scheme for preparing an arbitrary (including pure and mixed) qubit, where a partially entangled state and finite classical communication are used. To our knowledge, our scheme is the first RSP scheme that fits into this category. One other RSP scheme proposed by Berry shares close features, but can only be used to prepare an arbitrary pure qubit. Even so, our scheme saves classical communication by approximate 1 bit per prepared qubit under equal conditions. When using a maximally entangled state, the classical communication for our scheme is 2 bits, which agrees with Lo's conjecture on the resource cost. Furthermore Alice can switch between our RSP scheme and a standard teleportation scheme without letting Bob know, which makes the quantum channel multipurpose.
\end{abstract}
\maketitle

\section{Introduction}

Relying on quantum entanglement~\cite{horodecki_quantum_2009}, quantum communication protocols can present abilities that are unachievable by their classical counterparts. One notable example is quantum teleportation~\cite{bennett_teleporting_1993}. Using two entangled qubits and 2 bits of classical communication as resources, teleportation enables Alice to transmit an \textit{unknown} qubit state to Bob without physically transporting it. A variation of teleportation is remote state preparation (RSP)~\cite{lo_classical-communication_2000,pati_minimum_2000,bennett_remote_2001,devetak_low-entanglement_2001,luo_faithful_2012}, wherein Alice produces a \textit{known} qubit state at Bob's location by using entanglement and classical communication. Depending on the resource requirement, preparable state ensemble, and even successful rate, different kind of RSP scheme has been devised for carrying out the specific task.

In most RSP schemes, maximal entanglement, oftenest a Bell state, is assumed to be accessible. For example, the first RSP scheme proposed by Pati~\cite{pati_minimum_2000} uses a Bell state, realizing deterministic preparation of a qubit from a fixed great circle on the Bloch sphere. Other later proposed schemes extended the preparable ensembles to varying degrees for preparing more general qubits, like deterministic preparations of arbitrary pure qubits~\cite{hayashi_remote_2003,rosenfeld_remote_2007,liu_experimental_2007}, probabilistic preparations of arbitrary qubits~\cite{peters_remote_2005, liu_experimental_2007}, and more resently, schemes for deterministic preparations of arbitrary qubits were also implemented~\cite{killoran_derivation_2010, wu_deterministic_2010}. Unfortunately, these mentioned schemes are unadaptable to partially entangled resource states, which may occur in the real world. To devise RSP schemes that work with partial entanglement, new ways must be found.

Although it is more difficult to devise RSP schemes employing partial entanglement, Ye \textit{et al.}~\cite{ye_faithful_2004} proved that, at the expense of increased classical communication, it is possible to use a partially entangled state to perform deterministic RSP of an arbitrary pure state. Based on the proof, Berry~\cite{berry_resources_2004} proposed an explicit scheme for performing such an RSP. Thus the two-dimensional case of Berry's scheme provides a complete technique for preparing an arbitrary pure qubit. In Ref.~\cite{hua_scheme_2014}, we also proposed an optimization procedure that can be incorporated into Berry's scheme to reduce unnecessary classical communication. Despite all these efforts, to our best knowledge, as yet there is no deterministic RSP scheme that use partial entanglement to prepare arbitrary (include pure and mixed) qubit states, which we reported here.

In Sec.~\ref{secii}, by utilizing the ensembles that can be prepared using partial entanglement with only two bits of classical communication, we give our RSP scheme which trades off classical communication against entanglement to implement the preparation of an arbitrary qubit. Our results agree with Lo's conjecture on the resource cost for deterministic RSP. At the end of Sec.~\ref{secii} we briefly summarize our scheme, and show Alice can switch between this RSP and a standard teleportation without letting Bob know. In Sec~\ref{seciii}, we introduce a geometrical tool called Voronoi diagram for the calculation of the classical communication requirement for our scheme. The calculation shows our scheme saves approximate 1 bit of classical communication when compared with Berry's scheme for the preparation of an arbitrary pure qubit. The underlying cause of the resource saving will be explained. In Sec.~\ref{seciv}, we draw our conclusions.

\section{Deterministic remote preparation of an arbitrary qubit}
\label{secii}

In this section, we present a remote state preparation (RSP) scheme for Alice to prepare an arbitrary qubit state at Bob's location deterministically. The resources for performing this task are a partially entangled two-qubit state and finite bits of classical communication from Alice to Bob. The qubit to be prepared is known to Alice but unknown to Bob. Furthermore, we restrict the scheme to be oblivious, where the state to be prepared is known to Alice but unknown to Bob.

The initial setup of our RSP scheme is as follows. Assume Alice and Bob have shared an entangled resource state in the form
\begin{equation*}
\ket{r} =r _0\ket{00}_\text{AB}+r _1\ket{11}_\text{AB}\text{,}
\end{equation*}
or, equivalently,
\begin{equation}
\ket{r} =\cos  \frac{\theta _r}{2}\ket{00}_\text{AB}+\sin  \frac{\theta _r}{2}\ket{11}_\text{AB}\text{,}
\label{eqresource}
\end{equation}
where $0\leq r_1\leq r_0$ and $0\leq \theta_r\leq \pi$. Any entangled pure two-qubit state can be brought to this form via local unitary operations. The entanglement of $\ket{r}$ is quantified by the Von Neumann entropy of either of $\ket{r}$'s reduced states, namely,
\begin{equation}
E(\ket{r})=-r_0^2 \log r_0^2 -r_1^2 \log r_1^2\text{. }
\label{eqentanglement}
\end{equation} The scheme starts by a positive operator valued measurement (POVM) on Alice's system A. The measurement operators given by
\begin{multline}
\bigg\{M_m=p_m\left(\frac{1}{r_0}\ket{0}\bra{0}+\frac{1}{r_1}\ket{1}\bra{1}\right)\rho _m^T\\.\left(\frac{1}{r_0}\ket{0}\bra{0}+\frac{1}{r_1}\ket{1}\bra{1}\right)\bigg\}_{m=0}^3,
\label{eqpovm}
\end{multline}
where the superscript $T$ denotes the transposition, and the values of $p_m$ and $\rho _m$ are to be determined later. By implementing this POVM, Alice obtains a measurement outcome $m$ with probability calculated as $\bra{r}M_m\ket{r}=p_m$, and the corresponding state of Bob's post-measurement system will be $\text{tr}_\text{A}\left(M_m\ket{r}\bra{r}\right)/p_m=\rho _m$. Then, based on the obtained measurement result $m$, Alice instructs Bob to apply a corresponding unitary operation $U_m$ on his system by sending him 2 bits. The unitary operation $U_m$ is chosen from $\left\{I\text{,}\sigma _3\text{,}\sigma _1\text{,}-\sigma _3 \sigma _1\right\}$, which is a unitary opertation set duplicated purposely from a standard teleportation scheme, as we will discuss at the end of the section. Also, in order to transform Bob's system to a state deterministically, we make a convention that all $U_m \rho_m U_m^\dagger$ are equal.

Now, based on the above setup, we give the preparable ensemble of pure qubits in Theorem~\ref{theorem1}. This ensemble will be generalized to include mixed qubits by Corollary.~\ref{corollary1}.

\begin{theorem}
Using a preshared resource state $\ket{r} $ given by Eq.~(\ref{eqresource}) and 2 bits of classical communication from Alice to Bob, any pure qubit from the ensemble
\begin{multline*}
\bigg\{\ket{i(\theta_i,\phi_i)} =\cos  \frac{\theta _i}{2}\ket{0}+e^{i \phi _i}\sin  \frac{\theta _i}{2}\ket{1} \bigg\vert~\phi _i\in[0,2\pi)\\
\text{and }\theta _i\in[0,\theta _r]\cup[\pi -\theta _r,\pi]\bigg\}
\end{multline*}
can be remotely prepared. Particularly if $\ket{r}$ is a maximal entangled state, the above ensemble is represented by the entire Bloch sphere consisting of every pure qubit.
\label{theorem1}
\end{theorem}

\begin{proof}
Assume the pure qubit that Alice wants to prepare is expressed as
\begin{multline*}
\ket{i} =\cos  \frac{\theta _i}{2}\ket{0}+e^{i \phi _i}\sin  \frac{\theta _i}{2}\ket{1}\equiv i_0\ket{0}+e^{i \phi _i}i_1\ket{1}\text{, }\\0\leq \theta _i\leq \pi\text{, }0\leq \phi _i<2 \pi\text{.}
\end{multline*}
Here $\theta _i$ and $\phi _i$ are the polar and azimuthal angles of $\ket{i}$ in the Bloch sphere representation. For Bob, before he receives the 2 classical bits from Alice, his system is in the state
\begin{equation}
\begin{split}
\sum_{m=1}^4 p_m \rho_m&=r_0^2\ket{0}\bra{0}+r_1^2\ket{1}\bra{1}\text{.}
\end{split}
\label{eqsystemb}
\end{equation}
Substitute $\rho _m\equiv U_m^\dagger\ket{i}\bra{i}U_m$ with $U_m\in\left\{I\text{,}\sigma _3\text{,}\sigma _1\text{,}-\sigma _3 \sigma _1\right\}$ into Eq.~(\ref{eqsystemb}), we obtain
\begin{multline}
\left(i_0^2 P_1+i_1^2 P_2\right)\ket{0}\bra{0}+\left(i_0^2 P_2+i_1^2 P_1\right)\ket{1}\bra{1}\\=r_0^2\ket{0}\bra{0}+r_1^2\ket{1}\bra{1}\text{,}
\label{eqcondition}
\end{multline}
where $P_1\equiv 2 p_0=2 p_1$ and $P_2\equiv 2 p_2=2 p_3$. Moreover, for a legitimate POVM in Eq.~(\ref{eqpovm}), $P_1, P_2\geq 0$ and $P_1+P_2=1$ are implied. Eq.~(\ref{eqcondition}) can be used as the necessary and sufficient condition for checking the preparablity of a pure qubit $\ket{i}$ (also see Eq.~(3) in Ref.~\cite{ye_faithful_2004}). It is easy to see, as long as $r_0\leq \max  \left\{i_0,i_1\right\}$, Eq.~(\ref{eqcondition}) is soluble for non-negative $P_1$, $P_2$ with $P_1+P_2=1$. Except for $r_0=i_0=\frac{1}{\sqrt{2}}$, where $P_1$
can be any value in $[0,1]$, the universal solution to Eq.~(\ref{eqcondition}) is 
\begin{equation}
P_1=\left(r_0^2-i_0^2\right)/{2 \left(\frac{1}{2}-i_0^2\right)}.
\label{eqp1}
\end{equation}
For consistency, we only use the universal solution in the follow discussion.

On the Bloch sphere, the ensemble of states that satisfies $r_0\leq \max  \left\{i_0,i_1\right\}$ is represented by an antipodal pair of spherical caps with $\theta _i\in[0,\theta _r]\cup[\pi -\theta _r,\pi]$.
Particularly when $\ket{r}$ is maximally entangled, we have $\theta _i\in[0,\pi/2]\cup[\pi/2,\pi]$, which means the ensemble is represented by the entire Bloch sphere consisting of every pure qubit.
\end{proof}

A mixed qubit describes a two-dimensional quantum system whose state is not completely known. One can suppose such a system is in a pure qubit state $\ket{i}$ with probability $p$ and in the maximally mixed qubit state with probability $1-p$. The density matrix for such a mixed qubit can be expressed as
$$\rho _i(p,\theta_i,\phi_i)=p\ket{i(\theta_i,\phi_i)}\bra{i(\theta_i,\phi_i)}+(1-p)\frac{I}{2},$$
where
$$\ket{i(\theta_i,\phi_i)} =\cos  \frac{\theta _i}{2}\ket{0}+e^{i \phi _i}\sin  \frac{\theta _i}{2}\ket{1}.$$
Here $p$, $\theta _i$ and $\phi _i$ uniquely identify the Bloch vector $\boldsymbol{r}_i=(p\sin\theta_i\cos\phi_i,p\sin\theta_i\sin\phi_i,p\cos\theta_i)$ which is related to the position of $\rho_i$ in the Bloch ball.
When $p=1$, $\rho _i$ degenerates into a pure qubit $\ket{i}$, while from the pure $\ket{i}$ to the mixed $\rho _i$, the Bloch vector shrinks by a factor $p$. One can verify, if Alice sends Bob totally random classical bits (00, 01, 10 or 11 each with probability 1/4), Bob' s system will end up with the maximally mixed qubit represented by the density matrix $I/2$. Therefore, if Alice wants to remotely prepare $\rho _i$, she only needs to replace the bits used for preparing $\ket{i}$ by totally random bits with probability $1-p$.  Obviously we have the following corollary.

\begin{corollary}
Using a preshared resource state $\ket{r} $ given by Eq.~(\ref{eqresource}) and 2 bits of classical communication from Alice to Bob, any qubit state from the ensemble
\begin{multline*}
\bigg\{\rho _i(p,\theta_i,\phi_i)=p\ket{i(\theta_i,\phi_i)}\bra{i(\theta_i,\phi_i)}+(1-p)\frac{I}{2}\bigg\vert\\p\in[0,1], \phi _i\in[0,2\pi) \text{ and }\theta_i\in[0,\theta _r]\cup[\pi -\theta _r,\pi]\bigg\}
\end{multline*}
can be remotely prepared. Particularly if $\ket{r}$ is maximally entangled, the above ensemble is represented by the entire Bloch ball consisting of every pure and mixed qubit.
\label{corollary1}
\end{corollary}

Lo~\cite{lo_classical-communication_2000} conjectured that, with unlimited entanglement resource, deterministic preparation of an arbitrary pure qubit at Bob's location requires $2$ bits of classical communication. Bennet \textit{et al.}~\cite{bennett_remote_2001} have proved in a more restricted condition, where Bob is restricted to perform a unitary operation and is oblivious to the prepared state, the RSP must use at least 2 bits of classical communication. Our setup meets this condition, because when $\ket{r}$ is maximally entangled we have $P_1=1/2$ from Eq.~(\ref{eqp1}), which means Bob cannot extract from the classical communication any information about the prepared state (thus is oblivious). Now, by treating Theorem~\ref{theorem1} and Corollary~\ref{corollary1} as complements to Bennet \textit{et al.}'s proof, we know 2 bits of classical communication are both necessary and sufficient (even for preparing an arbitrary qubit).

Due to the strictly increasing relation between $\theta_r$ and $E(\ket{r})$, a less entangled $\ket{r}$ will lead to smaller preparable ensembles in Theorem~\ref{theorem1} and Corollary~\ref{corollary1}. In the following discussion, we show how to trade off classical communication for reduced entanglement to enable preparation of an arbitrary target qubit $\rho_t$. 

We use $C_1$ to denote the spherical cap lying at the north pole mentioned in Theorem~\ref{theorem1}, and $C_{-1}$ the antipodal one. The pair of $C_1$ and $C_{-1}$ is denoted by $C_{1,-1}$. The size of $C_{1,-1}$ can be directly measured by the parameter $\theta_r$. If $\ket{r}$ is partially entangled, then $C_{1,-1}$ will not cover the entire Bloch sphere. In order to prepare a target state $\rho_t$ with Bloch vector $\boldsymbol{r}_t$ outside the convex hull of $C_{1,-1}$, Alice and Bob need to proceed as follows. First, before the preparation begins, they need to determine a $K$-element rotation operation set $\left\{R_j\right\}_{j=1}^K$ that is dependant on $\ket{r}$. Then, Alice deliberately prepares the intermediate state $\rho_i=R_j^{-1}\rho_t R_j$, whose Bloch vector is inside the convex hull of $C_{1,-1}$. By sending Bob $\log{K}$ bits, Alice instructs him to use $R_j$ from $\left\{R_j\right\}_{j=1}^K$ to transforms his system into $\rho_t$.

The effect of a rotation operation on a state is to rotate the state's Bloch vector by a fixed angle about some axis of the Bloch ball. Let's suppose each $R_j$ maps the spherical cap $C_1$ to $C_j$ (and thus $C_{-1}$ to $C_{-j}$). Although the position of $C_{j,-j}$ may vary, all of them are the same size as $C_{1,-1}$ and the union of $\left\{C_{j,-j}\right\}_{j=1}^K$ must cover the entire Bloch sphere. As pointed out earlier, if $E(\ket{r})$ is reduced, the size of $C_{j,-j}$ will decrease too. This generally results in an increased classical communication cost as $K$ tends to become larger to ensure total coverage. The resource trade-off is inevitable for an RSP scheme of this type, but making $\left\{C_{j,-j}\right\}_{j=1}^K$ uniformly distributed can avoid overcommunication. The problem of how to uniformly distribute $\left\{C_{j,-j}\right\}_{j=1}^K$ can be rephrased as how to construct uniformly distributed $2K$ points with antipodal symmetry on the Bloch sphere, for we can use these points as the spherical caps' centers. The points construction method we use has been put into Appendix.

Now we summarize our scheme for remote preparation of an arbitrary qubit $\rho_t$ as the two-stage procedure below.

Stage 1. Using the POVM~(\ref{eqpovm}) and the unitary operations $I\text{, }\sigma _3\text{, }\sigma _1\text{, }-\sigma _3 \sigma _1$ to prepare an intermediate state $\rho_i=R_j^{-1}\rho_t R_j$ that belongs to the ensemble given in Corollary~\ref{corollary1}. The classical communication cost here is 2 bits. If $\ket{r}$ is maximally entangled, by setting $\left\{R_j\right\}_{j=1}^K\equiv\{I\}$, $\rho_t$ can be prepared within this stage.

Stage 2. If $\ket{r}$ is non-maximally entangled, Bob performs $R_j$ from the predefined set $\left\{R_j\right\}_{j=1}^K$ to transform $\rho_i$ to $\rho_t$. The classical communication cost in this stage is $\log{K}$ bits, which is traded off against $E(\ket{r})$.

Different from RSP, teleportation can only be carried out when a maximally entangled resource state is available~\cite{bennett_teleporting_1993}. In a standard teleportation scheme, where a Bell state $\ket{\Phi^+}\equiv\frac{1}{\sqrt 2}\ket{00}+\frac{1}{\sqrt 2}\ket{11}$ is used, Bob needs to apply a unitary operation chosen from $\left\{I\text{,}\sigma _3\text{,}\sigma _1\text{,}-\sigma _3 \sigma _1\right\}$ after he receives the outcome of Alice's Bell basis measurement. The same unitary operation set is used in our RSP scheme. As we have said, this choice is made on purpose, because one can see when $\ket{r}$ is maximally entangled, i.e., $\ket{r}=\ket{\Phi^+}$, regardless of the target state the probability that Bob uses any one of these unitary operations in both schemes is always 1/4. There is no chance for Bob to tell which scheme is being performed, Alice can switch between teleportation and RSP unilaterally.

\section{Classical communication cost}

\label{seciii}
The total classical communication cost for our scheme is $2+\log{K}$ bits. For a given $K$, the entanglement of the resource state $\ket{r}$ cannot below a certain lower bound, otherwise there always exist some unpreparable qubit states for our scheme. To calculate the lower bound of $E(\ket{r})$, we need to make use of a geometry tool called \textit{Voronoi diagram}. A Voronoi diagram is a partition of a space into regions based on distance to some specific points called sites. For each site, the corresponding region, called Voronoi cell, consists all points closer to this site than to any other. All Voronoi cells are polygon-shaped with edges equidistant from two sites and vertices equidistant from three or more sites. Fig.~\ref{fig1} gives an illustration of the Voronoi diagram generated from uniformly distributed 64 ($K=32$) points with antipodal symmetry.

\begin{figure}
	\includegraphics[width=0.33\textwidth]{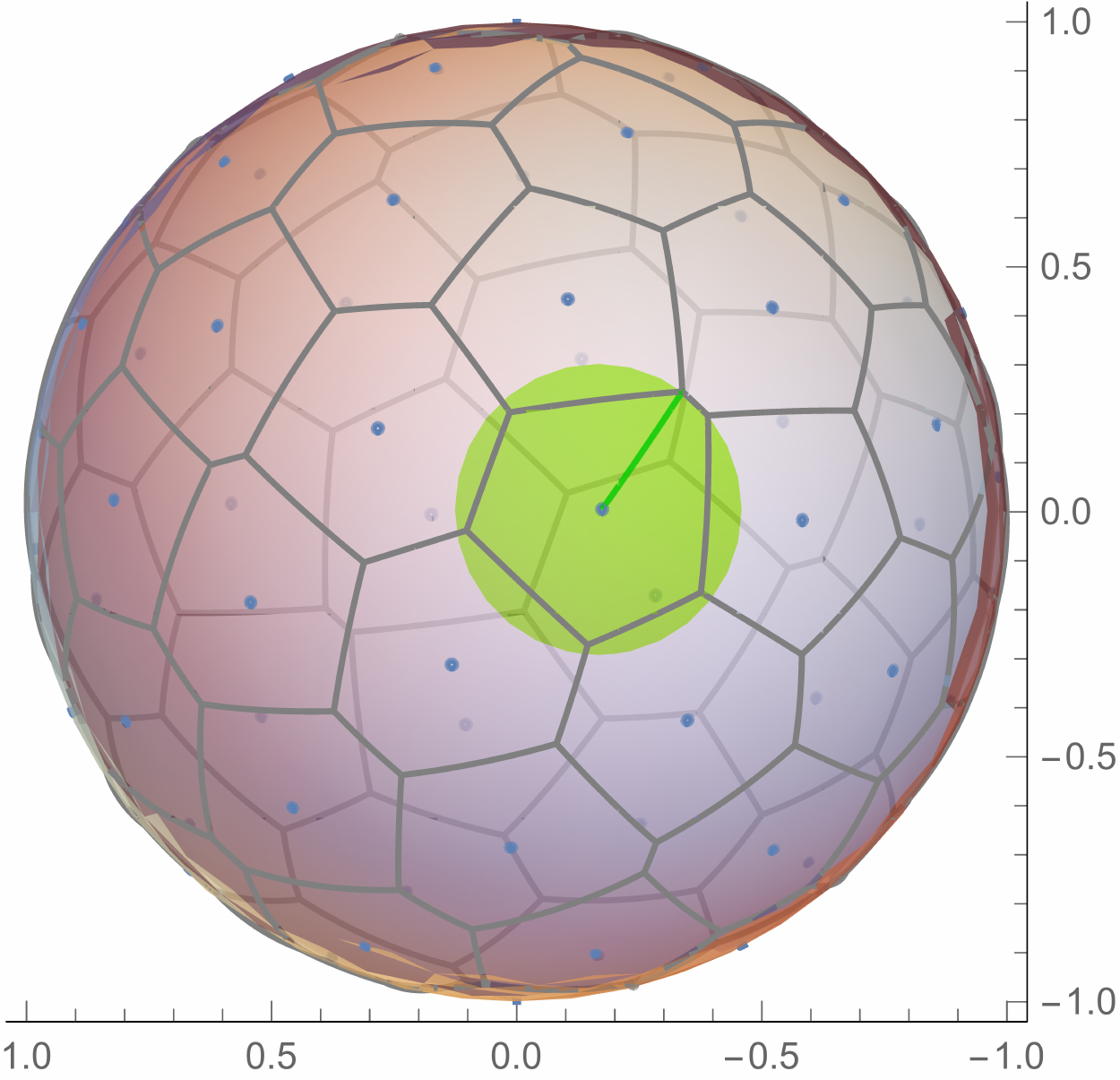}
	\caption{Voronoi diagram generated from uniformly distributed 64 points with antipodal symmetry on the Bloch sphere. The mesh on the sphere shows the Voronoi cells corresponding to thses points. The green segment is the longest site-vertex geodesic line. To cover the entire Bloch sphere, every $C_j$ must be no smaller than the green area.}
	\label{fig1}
\end{figure}

If we treat the centers of $C_j$, $j=\pm 1,...,\pm K$ as sites  denoted by $\boldsymbol{s}_j$, a Voronoi diagram can be generated. The necessary and sufficient condition for the union of all $C_j$'s to cover the entire Bloch sphere is that every $C_j$ covers the Voronoi cell corresponding to $\boldsymbol{s}_j$. We denotes by $\boldsymbol{v}_{j,k}$ the $k$th vertex of the Voronoi cell based on $\boldsymbol{s}_j$. After obtaining the coordinates of all $v_{i,j}$ numerically by computer~\cite{zheng_plane_2011}, the lower bound of $E(\ket{r})$ can be calculated from Eq.~(\ref{eqentanglement}) and $$\text{max}\big\{\arccos(\boldsymbol{v}_{j,k}\cdot\boldsymbol{s}_j)\big\vert\text{ for all valid }(j,k)\big\}\leq\theta_r.$$ Using the point sets from Appendix, the results for $K=2^n, n=1,2,...,10$ are both listed in Table~\ref{table1} and ploted in Fig.~\ref{fig2}.

\begin{table*}[ht]
\centering
\def\arraystretch{1.3}
\begin{tabular}{ >{\centering\let\newline\\\arraybackslash\hspace{0pt}}m{2cm}||>{\centering}m{1.4cm}|>{\centering}m{1.4cm}|>{\centering}m{1.4cm}|>{\centering}m{1.4cm}|>{\centering}m{1.4cm}|>{\centering}m{1.4cm}|>{\centering}m{1.4cm}|>{\centering}m{1.4cm}|>{\centering}m{1.4cm}|c}
\thickhline
$K$ & 2 & 4 & 8 & 16 & 32 & 64 & 128 & 256 & 512 & 1024\\
\hline
Total bits \newline cost & 3 & 4 & 5 & 6 & 7 & 8 & 9 & 10 & 11 & 12\\
\hline
Lower bound of $E(\ket{r})$ & 1 & 0.744008 & 0.502988 & 0.236295 & 0.155618 & 0.094967 & 0.056478 & 0.033252 & 0.018274 & 0.010069\\
\thickhline
\end{tabular}
\caption{For $K$ with a value no more than 128, we use the optimal point sets as input in calculation, while when $K=256,512,1024$ we use Koay's point sets (see Appendix). }
\label{table1}
\end{table*}

An earlier RSP scheme proposed by Berry~\cite{berry_resources_2004} can be used for preparing arbitrary pure qubits. We find it also fits into our two-stage procedure summarized in Sec.~\ref{secii}, just by replacing mixed $\rho_i$ and $\rho_t$ with pure $\ket{i}$ and $\ket{t}$, respectively. The main difference is, in Berry's scheme, the intermediate state ensemble is represented by the spherical cap on Bloch sphere given by
\begin{equation*}
\bigg\{\ket{i} =\cos  \frac{\theta _i}{2}\ket{0}+e^{i \phi _i}\sin  \frac{\theta _i}{2}\ket{1} \bigg\vert~\theta _i\in[0,\theta _r]\bigg\},
\end{equation*}
which is only half-size of the intermediate state ensemble in the Theorem~\ref{theorem1}. So two classical bits of communication will not be sufficient for Berry's scheme for preparing an arbitrary pure qubit, even when the resource state is maximally entangled. The double-sized intermediate state ensemble in our scheme reduces the number of elements in devising the unitary operation set $\left\{R_j\right\}_{j=1}^K$ to nearly half and eventually cut down the total classical communication cost by approximate 1 bit per pure qubit. The approximation is caused by the different symmetries employed in constructing uniformly distributed spherical caps in these two schemes. In Fig.~\ref{fig2}, we include the result from~\cite{hua_scheme_2014} for comparison.

\begin{figure}
	\includegraphics[width=0.47\textwidth]{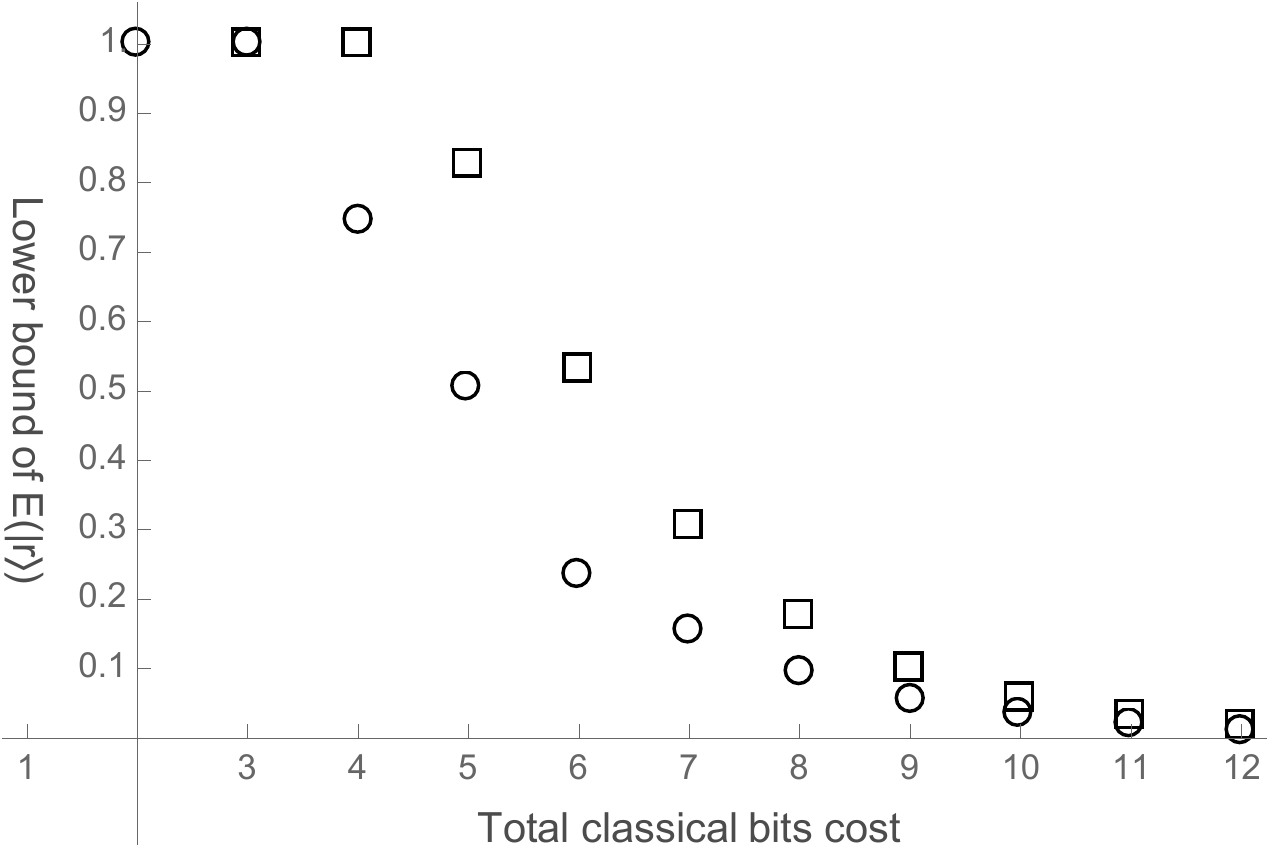}
	\caption{The total classical bits cost versus the lower bound of $E(\ket{r})$ for RSP of general qubits. The circles represent the required resource in the RSP scheme proposed in this paper for preparing an arbitrary (including pure and mixed) qubit. The squres represent the result of Berry's scheme (after optimization) for preparing an arbitrary pure qubit based on the data from Ref.~\cite{hua_scheme_2014}.}
	\label{fig2}
\end{figure}

\section{Conclusions}
\label{seciv}
We have proposed an RSP scheme for remotely preparing a general qubit by using any pure entangled state and finite classical bits. Our scheme can be treated as a two-stage procedure. If a maximally entangled resource state is available, the target qubit can be directly prepared in the first stage with 2 bits of classical communication, which agrees with Lo's conjecture on the resource cost for deterministic RSP. If the resource state is only partially entangled, an additional rotation operation will be performed in the second stage to transform the intermediate state prepared in the first stage to the final target state. The total classical communication cost is shown to be traded off against the resource state's entanglement. To the best of our knowledge, our scheme is the first deterministic RSP scheme for preparing an arbitrary qubit using a partially entangled state and finite classical communication. Theoretically, Our technique can be generalized to a higher dimension, but the geometry of qudits $(d>2)$ may be hard to deal with.

Our scheme also shares the same unitary operation set with the standard teleportation scheme. The benefit is when the resource state $\ket{r}$ is maximally entangled, Alice can switch between teleportation and RSP without letting Bob know, because no matter in which scheme Bob always performs a unitary operation $U_m\in\left\{I\text{,}\sigma _3\text{,}\sigma _1\text{,}-\sigma _3 \sigma _1\right\}$ with probability 1/4. This feature can make an entangled channel more versatile without sacrificing flexibility.

\section*{Acknowledgments}
We thank Cheng Guan Koay and Lin Chen for valuable discussion. We also gratefully acknowledge the support by NNSF of China, Grant No. 11375150.

\section*{Appendix: Uniform distribution of antipodally symmetric points on the unit sphere}
On a sphere, point sets with antipodal symmetry have special importance in both scientific and engineering fields, many works has been published for generating such sets. The methods for generating a $2K$-element set with antipodal symmetry usually contain a minimization procedure of electrostatic potential energy. For the number of elements within a few hundreds, the point sets are tabulated online~\cite{_optimal_????}. However, for the number of points in these point sets beyond a few hundreds, the optimization procedure will become unwieldy. To solve this problem, we can use instead some constructive methods to generate nearly uniform point sets with antipodal symmetry, which give very close results especially when $K$ is large. In this work for generating antipodally symmetric point sets with $K>256$ , we use a simple deterministic points construction scheme proposed by Koay~\cite{koay_simple_2011}. For the unit sphere, the spherical coordinates $\left(1,\theta _i,\phi _{i,j}\right)$ of the points on the upper hemisphere is given by:
\begin{align*}
\theta_i&=(i-\frac{1}{2})\frac{\pi }{[N]},           &  i&=1,2,...,[N],\\
\phi _{i,j}&=(j-\frac{1}{2})\frac{2\pi }{K_i},    &  j&=1,2,...,K_i.
\end{align*}
where $N$ is the solution to $N=\frac{K}{2}\sin{\frac{\pi}{4N}}$, $[\cdot]$ is the function which gives the integer closest to the input, and
\begin{equation*}
K_i=
\begin{cases}
\left[\frac{2 \pi\sin  \theta _i}{\pi\csc{\frac{\pi}{4[N]}}}K\right],  &i=1,2,...,[N]-1,\\
K-\sum_{i=1}^{[N]-1}K_i,  &i=[N].
\end{cases}
\end{equation*}

\end{document}